\documentclass[10pt, twocolumn]{IEEEtran}
\usepackage{pifont}
\usepackage{cite}
\usepackage[dvips]{epsfig}
\usepackage{graphicx}
\usepackage{calligra}
\usepackage[T1]{fontenc}
\usepackage{bbold}
\usepackage{mathrsfs}

\usepackage{subfigure}
\usepackage{color}
\usepackage{amsmath}
\usepackage{cases}

\usepackage{amssymb}
\usepackage[justification=centering]{caption}
\usepackage{framed}

\usepackage{subfigure}
\usepackage{color}
\usepackage{amsmath}
\usepackage{bm}
\usepackage{amssymb}
\usepackage{amsbsy}
\usepackage{bbm}
\usepackage{dsfont}
\usepackage[normalem]{ulem}
\usepackage{xcolor,cancel}
\usepackage{enumerate}
\usepackage{paralist}

\usepackage{algorithm, algorithmic}

\usepackage{amsthm}
\usepackage{fancyhdr}
\usepackage{caption}

\newtheorem{theorem}{Theorem}

\newtheorem{condition}{Condition}

\newtheorem{corollary}{Corollary}

\newtheorem{definition}{Definition}

\newtheorem{proposition}{Proposition}

\newtheorem{lemma}{Lemma}

\newtheorem{remark}{Remark}

\newtheorem{claim}{Claim}

\newtheorem{assumption}{Assumption}

\newcommand{\cL}{\mathcal{L}}

\newcommand{\cS}{\mathcal{S}}

\newcommand{\bs}{{\bf{s}}}

\newcommand{\bdelta}{{\boldsymbol{\delta}}}

\newcommand{\ignore}[1]{{}}

\newcounter{parentalgorithm}

\ifCLASSINFOpdf
\else
\fi

\begin{document}
%
\title{Time-Restricted Double-Spending Attack on PoW-based Blockchains}
%
%
%

\author{Yiming Jiang  and Jiangfan Zhang, ~\IEEEmembership{Member,~IEEE}
\thanks{Y. Jiang and J. Zhang are with the Department of Electrical and Computer Engineering, Missouri University of Science and Technology, Rolla MO 65409 USA (e-mail: yjk7z@mst.edu, jiangfanzhang@mst.edu)} 
}



\maketitle
%
\begin{abstract}
Numerous blockchain applications are designed with tasks that naturally have finite durations, and hence, a double-spending attack (DSA) on such blockchain applications leans towards being conducted within a finite timeframe, specifically before the completion of their tasks. Furthermore, existing research suggests that practical attackers typically favor executing a DSA within a finite timeframe due to their limited computational resources. These observations serve as the impetus for this paper to investigate a time-restricted DSA (TR-DSA) model on Proof-of-Work based blockchains. In this TR-DSA model, an attacker only mines its branch within a finite timeframe, and the TR-DSA is considered unsuccessful if the attacker’s branch fails to surpass the honest miners’ branch when the honest miners’ branch has grown by a specific number of blocks. First, we developed a general closed-form expression for the success probability of a TR-DSA. This developed probability not only can assist in evaluating the risk of a DSA on blockchain applications with timely tasks, but also can enable practical attackers with limited computational resources to assess the feasibility and expected reward of launching a TR-DSA. In addition, we provide rigorous proof that the success probability of a TR-DSA is no greater than that of a time-unrestricted DSA where the attacker indefinitely mines its branch. This result implies that blockchain applications with timely tasks are less vulnerable to DSAs than blockchain applications that provide attackers with an unlimited timeframe for their attacks. Furthermore, we show that the success probability of a TR-DSA is always smaller than one even though the attacker controls more than half of the hash rate in the network. This result alerts attackers that there is still a risk of failure in launching a TR-DSA even if they amass a majority of the hash rate in the network.

\end{abstract}


\IEEEpeerreviewmaketitle

\section{Introduction}
As the market price of cryptocurrencies has surged in recent years, the underlying technology known as blockchain has attracted a growing interest. Blockchain, a chronologically ordered sequence of blocks that are cryptographically linked to each other, revolutionizes the way information is secured, distributed, and shared. While blockchain was initially recognized primarily for its role in financial applications such as cryptocurrencies \cite{nakamoto2008bitcoin, wood2014ethereum,hopwood2016zcash}, its inherent security features and decentralized nature have broadened its utility across diverse engineering fields  \cite{kurt2019secure,yang2018blockchain,sharma2018distarch,jiang2023distributed,li2022blockchain,su2022lvbs,weng2021deepchain}. 

Specially, blockchain can serve as a distributed database that safeguards stored data, which employs a consensus protocol to ensure that all participants within the network maintain identical data copies. There are many types of blockchain consensus protocols used in blockchain, such as Proof-of-Work (PoW), Proof-of-Stake, Proof-of-Activity, and Proof-of-Capacity. Among these, the PoW is one of the most widely used consensus protocols. For example, Bitcoin, Litecoin, ZCash, and Bitcoin Cash all adopt PoW as their consensus protocols \cite{ouyang2021pow,hopwood2016zcash,kwon2019bitcoin}. In view of the popularity of the PoW consensus protocol, we focus on PoW-based blockchains in this paper. 

Since the emergence of the first PoW-based blockchain application, Bitcoin, a significant amount of work has been conducted to study the security of the PoW-based blockchains. For instance, some studies \cite{garay2015bitcoin,pass2017analysis,kiffer2018better,ren2019analysis,dembo2020everything} have analyzed essential properties of PoW-based blockchains, such as consistency, ensuring that all honest parties maintain the same blockchain copy during PoW execution. Some other studies focus on investigating security threats to PoW-based blockchains, including double-spending attack (DSA)\cite{nakamoto2008bitcoin,rosenfeld2014analysis,zhang2019double,grunspan2018double,ozisik2017explanation}, selfish mining\cite{eyal2018majority,bai2019deep}, pool-hopping attack \cite{shi2020fee}, routing attack\cite{apostolaki2017hijacking}, eclipse attack\cite{heilman2015eclipse,yves2018total}, and sybil attack\cite{zhang2019double}. Among these attacks, DSA is one of the most destructive threats to blockchain systems, and it has already resulted in significant losses across numerous PoW-based blockchain applications. For example, one of Bitcoin forks, Bitcoin Gold, suffered double-spending attacks in 2018, and again in 2020, with more than 17 million U.S. dollars lost in total \cite{bitcoingold}. Zencash fell victim to a DSA in 2018, which caused losses exceeding \$460,000 based on the cryptocurrency price at that time \cite{zencash}. In 2020, an attacker successfully double-spent 238,306 ethereum classic tokens, valued at \$1.68 million \cite{ethereumclassic}. Verge encountered a significant block reorganization involving more than 560,000 blocks due to a DSA in 2021, which led to the erasure of 200 days of transactions \cite{verge}.

A PoW-based blockchain determines its main chain based on the concept of the longest chain rule \cite{nakamoto2008bitcoin,bashir2017mastering,antonopoulos2014mastering}. In other words, the main chain of a PoW-based blockchain is the longest branch within it. The DSA on a PoW-based blockchain manifests as an attempt to secretly mine a branch within the blockchain with the objective of making its own branch the longest. If a DSA launched by an attacker is successful, i.e., the attacker mines its branch to become the longest branch within a blockchain, the attacker can alter data stored in the blocks of the main chain due to the PoW consensus protocol \cite{nakamoto2008bitcoin,bashir2017mastering,zaghloul2020bitcoin}. In many previous works \cite{yang2021priscore,hjalmarsson2018blockchain,lei2017blockchain,ling2019blockchain,li2018creditcoin,zhou2018beekeeper,
yang2021privacy,asefi2021application}, data stored in a blockchain are commonly referred to as transactions even for non-financial blockchain applications. In this paper, we adopt the same terminology and use the term ``transactions'' to refer to data stored in a blockchain.  We refer to miners that follow the standard PoW consensus protocol as honest miners which mine transactions into blocks to form an honest branch. To mitigate DSAs (i.e., to mitigate the risk of a transaction being altered by an attacker), a transaction is only considered confirmed if it has been included in a block of the longest branch and a designated number of subsequent blocks have been linked to it (typically six blocks in the case of Bitcoin). However, this kind of confirmation protocol cannot completely prevent DSAs against PoW-based blockchains. To be specific, a successful DSA can be achieved by an attacker which secretly mines a branch, containing a transaction that conflicts with a target transaction, to catch up with and surpass the honest branch after the target transaction has been confirmed in the honest branch. This kind of DSA issue has motivated a lot of research in prior studies, see \cite{ozisik2017explanation,zaghloul2020bitcoin, rosenfeld2014analysis,grunspan2018double,zhang2019double,grunspan2022profitability} for instance.  In these prior studies, when mathematically characterizing a DSA against a PoW-based blockchain, an infinite Gambler's Ruin model is commonly adopted to describe the catch-up process of the attacker's branch. Specifically, this infinite Gambler's Ruin model assumes that the attacker can indefinitely mine its own branch to catch up with and surpass the honest branch \cite{ozisik2017explanation,zaghloul2020bitcoin, rosenfeld2014analysis,grunspan2018double,zhang2019double,grunspan2022profitability}. We use the term ``time-unrestricted DSA (TU-DSA)'' to refer to the DSA where the attacker can indefinitely mine its own branch to catch up with and surpass the honest branch. However, this TU-DSA model is not appropriate for many blockchain applications. 

For example, the infinite Gambler's Ruin model is not suitable for modeling DSAs on a certain type of blockchain application. This is because these blockchain applications require the completion of their tasks within finite timeframes by utilizing historical and present data stored in their blockchains, see \cite{su2022lvbs,yang2021priscore,hjalmarsson2018blockchain,hassija2020traffic,lei2017blockchain, ling2019blockchain,li2018creditcoin,zhou2018beekeeper,yang2021privacy,li2022blockchain,asefi2021application}  and the references therein. To be specific, the task of this type of blockchain application is executed and accomplished when the longest branch within its blockchain has grown by a specific number of blocks, signifying the accumulation of sufficient data in that branch. It is worth mentioning that once this type of blockchain application's task is accomplished, any future data stored in the blockchain or any future falsification of the historical and present data stored in the blockchain will not affect the accomplished task. Consequently, any DSA on this type of blockchain application should be regarded as unsuccessful if the attacker fails to mine its branch to surpass the honest branch before the application task is completed. To this end, for a DSA on this type of blockchain application, it is more appropriate to assume that the attacker only mines its branch within a finite timeframe, specifically before the completion of the application task. In other words, if the attacker's branch remains shorter than the honest branch, the attacker only mines its branch until the honest branch has grown by a specific number of blocks. 

Furthermore, the assumption that an attacker can indefinitely mine its branch to catch up with and surpass the honest branch is impractical. This is due to the fact that block mining in PoW-based blockchains requires significant computational power, which in turn consumes a substantial amount of electricity. If an attacker persists in mining its branch indefinitely, the attacker may have to expend a vast amount of electricity and resources, particularly when it takes an extensive period for the attacker to extend its branch to surpass the honest branch. In light of this, from a practical standpoint, an attacker should refrain from indefinitely investing its computational power and resources in performing a single DSA. Instead, it is advisable for an attacker to restrict the timeframe for conducting a DSA, thereby reducing the potential risk of excessive resource consumption on a single DSA. Notably, this viewpoint finds support in recent literature \cite{ozisik2017explanation,zaghloul2020bitcoin,zheng2023adaptive}. 
As the average block time of a PoW-based blockchain is typically constant \cite{zaghloul2020bitcoin,biais2019blockchain,antonopoulos2014mastering,bashir2017mastering}, the timeframe for conducting a DSA can be effectively restricted by leveraging the growth of the honest branch. Specifically, the attacker can restrict the timeframe for conducting a DSA by ensuring that it concludes when the honest branch has grown by a specific number of blocks. As such, if the attacker fails to mine its branch to surpass the honest branch before the honest branch has grown by a specific number of blocks, then the attacker should stop its DSA, and the DSA is considered unsuccessful.

Taking into account the aforementioned considerations, we consider a time-restricted DSA (TR-DSA) on a PoW-based blockchain in this paper. The TR-DSA is regarded as unsuccessful if the attacker's branch fails to surpass the honest branch before the honest branch has grown by a specific number of blocks. This TR-DSA model is well suited to describe DSAs on many blockchain applications particularly those where the application tasks must be accomplished within finite timeframes. Additionally, this TR-DSA model is appropriate for describing DSAs for scenarios where a practical attacker cannot indefinitely invest its computational power and resources in conducting a single DSA.

\subsection{Summary of Results and Main Contributions}

We consider a TR-DSA model in which a TR-DSA is deemed successful if the attacker's branch successfully surpasses the honest branch before the honest branch has grown by a certain number of blocks. We develop a closed-form expression for the success probability of a TR-DSA. To achieve this, we first show that for a TR-DSA, the process of the attacker's branch surpassing the honest branch can be cast as a two-sided boundary hitting problem for a two-dimensional random walk with two possible walking directions. In particular, the two boundaries in this boundary hitting problem are orthogonal to each other, while the two possible walking directions of the two-dimensional random walk are not orthogonal to each other. Moreover, one walking direction of the random walk is neither orthogonal nor parallel to any of the two boundaries. For such a two-sided boundary-hitting problem, a general closed-form expression for the probability of the random walk hitting one boundary before hitting the other is derived, which is then utilized in the development of the closed-form expression for the success probability of a TR-DSA. 

The developed closed-form expression for the success probability of a TR-DSA is useful, primarily manifested in the following two aspects. On one hand, for blockchain applications where their tasks have to be accomplished within finite timeframes, the developed closed-form expression can serve as a valuable tool for assessing their vulnerability to double-spending attacks. Additionally, these blockchain applications can leverage this closed-form expression to improve their block confirmation protocol, thereby effectively mitigating the risk of a TR-DSA. On the other hand, for potential attackers with limited computational resources, the developed closed-form expression for the success probability of a TR-DSA can be employed to evaluate the expected rewards associated with launching a TR-DSA based on the attacker's available computational resources. This evaluation can aid potential attackers in deciding whether launching a TR-DSA is a worthwhile endeavor.

Furthermore, by leveraging the developed closed-form expression for the success probability of a TR-DSA, we conduct a theoretical comparison between the success probability of a TR-DSA and the success probability of a TU-DSA. To be specific, we prove that the probability of success in launching a TR-DSA  is no greater than that in launching a TU-DSA. In particular, the success probability of a TR-DSA  is strictly smaller than that of a TU-DSA when the probability of next mined block in the blockchain being generated by the attacker is between zero and one. In addition, we show that the success probability of a TR-DSA is strictly smaller than one even though the probability of next mined block in the blockchain being generated by the attacker is between 0.5 and one. These results have two significant implications. On one hand, the developed results demonstrate that the blockchain applications whose tasks have to be accomplished within finite timeframes are inherently less vulnerable to double-spending attacks than the blockchain applications which provide attackers with unlimited timeframes for their attacks. On the other hand, the developed results alert attackers with limited computational resources that there is still a risk of failure in launching a TR-DSA even if they amass a majority of hash rate in the network.


\subsection{Related Work}

Blockchain technology has seen widespread adoption across diverse fields for securing data exchanges and storage in decentralized systems, see \cite{su2022lvbs,yang2021priscore,hjalmarsson2018blockchain,hassija2020traffic,lei2017blockchain, ling2019blockchain,li2018creditcoin,zhou2018beekeeper,yang2021privacy,li2022blockchain,aitzhan2018security,weng2021deepchain,tian2022blockchain,li2023astraea} and the references therein. For example, the authors in \cite{aitzhan2018security} implemented a blockchain-based energy trading system within a smart grid, leveraging blockchain to ensure secure energy trading transactions without the need for trusted third parties. A decentralized secure auditing framework based on blockchain was designed in \cite{li2023astraea} for donation systems, where the blockchain effectively eliminates a single point of failure and provides verifiable and reliable records of donations. It is noteworthy that all these works emphasize that DSAs pose a non-neglected threat to blockchain applications, and addressing them is essential for maintaining the security and trustworthiness of data stored in these blockchain-based systems.

The scenarios where a practical attacker cannot indefinitely invest its computational resources in conducting a DSA have been considered in prior literature \cite{ozisik2017explanation,zaghloul2020bitcoin,jang2020profitable,grunspan2022profitability,zheng2023adaptive}. In \cite{ozisik2017explanation}, the authors pointed out that the assumption made in \cite{nakamoto2008bitcoin}, where an attacker can indefinitely invest its computational resources to mine its branch, has its limitations since an attacker's computational resources are limited in practice. In \cite{zaghloul2020bitcoin}, in order to ensure profitability, the authors suggested limiting the timeframe of a DSA. However, these studies do not theoretically investigate DSAs that are not conducted indefinitely.
To prevent attackers from wasting computational resources on a DSA with a low probability of success,
the authors of \cite{grunspan2022profitability} proposed an $A$-nakamoto DSA strategy, where an attacker terminates its DSA once the attacker's branch lags by $A$ blocks behind the honest branch. This $A$-nakamoto DSA strategy allows an attacker to stop its DSA early for some cases. However, for some other cases, this $A$-nakamoto DSA strategy may still require an attacker to indefinitely invest its computational resources in a single DSA. For example, consider the case where an attacker's branch and the honest branch happen to grow by one block alternately. For this case, the attacker still mines its branch indefinitely according to the $A$-nakamoto DSA strategy. Therefore, the $A$-nakamoto DSA strategy cannot address the practical concerns considered in this paper. 
In \cite{zheng2023adaptive}, the authors proposed an adaptive DSA strategy which allows an attacker to terminate its DSA within a finite timeframe. The motivation for this adaptive DSA strategy is to prevent an attacker from a large loss caused by a failed DSA under unfavorable situations that the success probability of the DSA is very low. However, no analysis on the success probability of the proposed adaptive DSA is conducted in \cite{zheng2023adaptive}. In this paper, we consider a related yet distinct DSA model, TR-DSA, and the primary focus of our work is to theoretically develop the success probability of the TR-DSA. Additionally, in \cite{jang2020profitable}, the authors considered a type of DSA which is only conducted within a finite timeframe as well. To be specific, the timeframe for this type of DSA concludes when the total number of newly mined blocks in both the honest branch and the attacker's branch reaches a designated number. This differs from the TR-DSA model considered in this paper, and therefore, the results in \cite{jang2020profitable} cannot be applied to the TR-DSA considered in this paper. In addition, it is worth mentioning that the DSA model considered in \cite{jang2020profitable} is not suitable for modeling DSAs on blockchain applications where the completion of application tasks depends on the accumulation of a substantial amount of data in the longest branch within the blockchain, specifically when the length of the longest branch reaches a designated number. Notably, the investigation of DSAs on this particular type of blockchain application is one of primary focuses of this paper.


A number of previous works have studied the success probability of a TU-DSA on a PoW-based blockchain. Nakamoto's seminal work \cite{nakamoto2008bitcoin} initially utilized an infinite Gambler's Ruin model to derive the probability of successfully launching a DSA. Further comprehensive explanations of this probability can be found in \cite{ozisik2017explanation} and \cite{zaghloul2020bitcoin}.  Later, \cite{rosenfeld2014analysis} and \cite{grunspan2018double} developed more accurate expressions for the success probability of a TU-DSA. Additionally, \cite{zhang2019double} calculated the success probability of a TU-DSA while taking network delays into account. All these prior works employ the infinite Gambler's Ruin model to characterize the catch-up process of the attacker's branch where the attacker indefinitely mines its own branch with the goal of surpassing the honest branch. However, this infinite Gambler's Ruin model cannot be employed to characterize the TR-DSA model considered in this paper.

The paper is organized as follows. Section \ref{Section_systemmodel} briefly introduces the time-restricted double-spending attack model. The success probability of a TR-DSA on a PoW-based blockchain is analyzed in Section \ref{Section_trdsa}. The comparison between the success probability of TR-DSA with that of a TU-DSA is conducted in Section \ref{Section_vs}. Numerical simulations are provided in Section \ref{Section_simulation}, and Section \ref{Section_conclusion} provides our conclusions.

\section{Probability of Success in Launching A Time-Restricted Double-Spending Attack}\label{Section_Probability_DSA}
\subsection{Time-Restricted Double-Spending Attack Model}\label{Section_systemmodel}

In this subsection, we briefly introduce the TR-DSA model considered in this paper. The mechanism of the TR-DSA is similar to that of the TU-DSA. We refer to \cite{ozisik2017explanation} for details of the mechanism of the TU-DSA. The primary distinction between the TR-DSA and the TU-DSA is that an attacker launching a TR-DSA stops mining its branch if its branch fails to surpass the honest branch when the honest branch has grown by a specific number of blocks. Conversely, an attacker launching a TU-DSA does not halt mining its branch under similar circumstances.

We assume that there is a portion of miners, called malicious miners, which
are under the command of an attacker. The honest miners are not controlled by the attacker and follow the standard PoW consensus protocol to mine blocks. The attacker's goal is to double-spend a transaction $TX_1$. In a PoW-based blockchain, a successful double-spending of $TX_1$ occurs when $TX_1$ is confirmed in the current longest branch and another longer branch that includes a conflicting transaction $TX_2$ emerges subsequently. To achieve this, the attacker initiates $TX_1$ and broadcasts it to all the miners. Upon receiving $TX_1$, the honest miners start mining it into a block stored in an honest branch. As the honest miners keep mining new blocks, the honest branch grows, and $TX_1$ will be stored and confirmed in the honest branch. The confirmation of $TX_1$ requires a designated number, denoted as $Z$, of subsequent blocks have been mined and added to the honest branch after the block containing $TX_1$. While waiting for the confirmation of $TX_1$, the attacker commands the malicious miners to secretly mine its own branch including $TX_2$ that follows the current latest block of the honest branch when $TX_1$ is broadcasted to the honest miners, as depicted in Fig. \ref{fig_DSA}. 

When $TX_1$ is confirmed in the honest branch, if the attacker's branch has already surpassed the honest branch, the attacker broadcasts its own branch to the honest miners. As such, all the honest miners switch to mining the attacker's branch according to the PoW consensus protocol, and $TX_2$ is included in the longest branch within the blockchain, successfully replacing $TX_1$. This successful replacement implies that the attacker has double-spent $TX_1$ successfully. If the attacker's branch has not surpassed the honest branch when $TX_1$ is confirmed, the attacker continues mining its branch with the aim of extending its branch to be longer than the honest branch. For a TR-DSA, when the honest branch has grown by $L$ more blocks after $TX_1$ has been confirmed, the attacker stops mining its branch if the attacker's branch does not surpass the honest branch. In other words, a TR-DSA is deemed unsuccessful if the attacker's branch fails to surpass the honest branch when $L$ additional blocks have been added to the honest branch after the confirmation of $TX_1$. 
\captionsetup[figure]{labelsep=period}
\begin{figure*}
	\centering   
	\includegraphics[width=1.0 \textwidth]{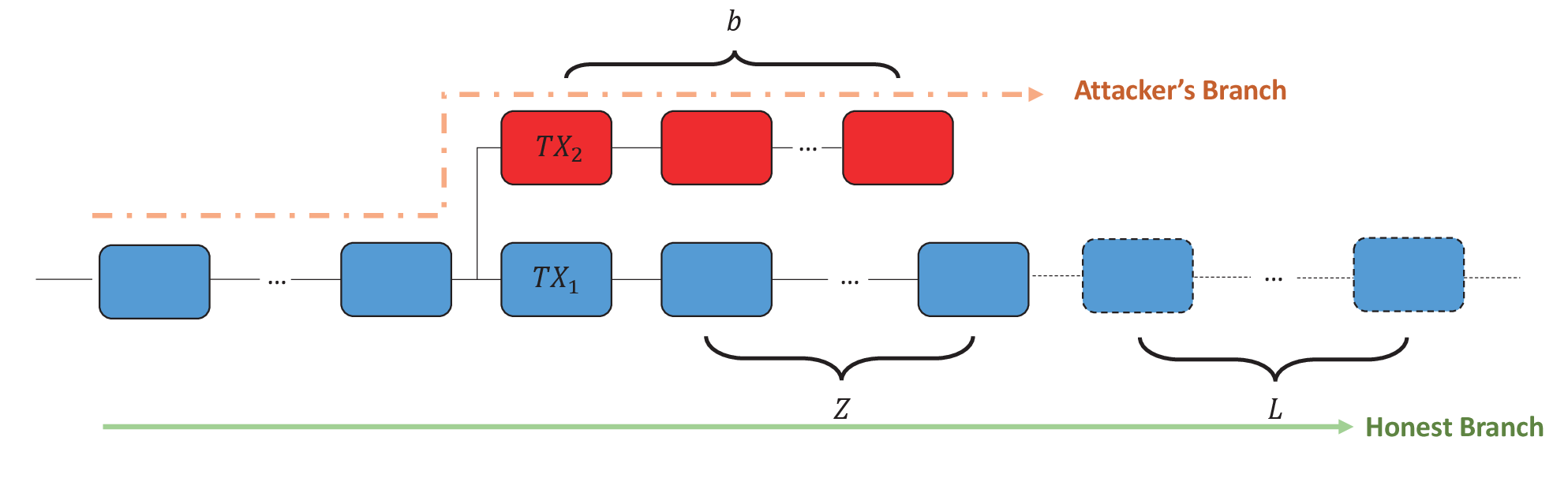}

	\caption{Illustration of a time-restricted double-spending attack.}
	\label{fig_DSA}
\end{figure*}
\subsection{The Success Probability of a Time-Restricted Double-Spending Attack}\label{Section_trdsa}
In this subsection, we derive the success probability of a TR-DSA.
Let $I$ represent the probability that the next block mined in the blockchain is generated by the malicious miners, and therefore, the probability that the next block mined in the blockchain is generated by the honest miners is $1-I$. It is worth mentioning that the scenario where $I=1$ is trivial. In such a scenario, the attacker completely controls the growth of the blockchain and can effortlessly perform double spending on any transactions stored in the blockchain with a guaranteed success probability of one. Thus, we only consider non-trivial scenarios where $0\le I<1$ in the subsequent analysis.

Let $b$ denote the number of blocks mined by the attacker in the attacker's branch when $TX_1$ is confirmed, as illustrated in Fig. \ref{fig_DSA}. There are two possible cases regarding the relationship between $b$ and $Z$. For the case where $b>Z+1$, the attacker's branch is longer than the honest branch when $TX_1$ is confirmed. As such, the attacker can broadcast its branch to all the honest miners to double spend $TX_1$ successfully, which implies that the TR-DSA succeeds. For the other case where $b\le Z+1$, in order to successfully double spend $TX_1$, the attacker has to continue mining its branch after the confirmation of $TX_1$. If the attacker's branch surpasses the honest branch before the honest branch grows by $L$ more blocks after $TX_1$ has been confirmed, the TR-DSA succeeds. Otherwise, the TR-DSA is deemed unsuccessful. Therefore, the success probability of a TR-DSA, $P_s^{(TR)}$, can be expressed as
\begin{align}\label{equ_trdsa} \notag
P_s^{(TR)}=&\Pr\{b> Z+1|\mathcal{E}\}\\
&+\sum_{k=0}^{Z+1}\Pr\{b=k|\mathcal{E}\}P_L(Z+1-k),
\end{align}
where $\mathcal{E}$ stands for the event that $TX_1$ is confirmed, that is, $Z$ blocks are mined and added to the honest branch after the block containing $TX_1$. $\Pr\{b> Z+1|\mathcal{E}\}$ represents the conditional probability that the attacker's branch surpasses the honest branch when $TX_1$ is confirmed. $\Pr\{b=k|\mathcal{E}\}$ is the conditional probability that the attacker's branch grows by $k$ blocks when $TX_1$ is confirmed. $P_L(Z+1-k)$ denotes the probability of the event where the attacker's branch lags by $(Z+1-k)$ blocks behind the honest branch when $TX_1$ is confirmed, and the attacker extends its branch to be longer than the honest branch before the honest branch grows by $L$ blocks after $TX_1$ has been confirmed.

The probability $\Pr\{b=k|\mathcal{E}\}$ has been well studied in previous works \cite{rosenfeld2014analysis,grunspan2018double} which show that $b$ follows a negative binomial distribution. To be specific,
\begin{align}\label{equ_bio}
\Pr\{b=k|\mathcal{E}\}=I^k(1-I)^{Z+1}{k+Z \choose k}.
\end{align} 
In addition, $\Pr\{b> Z+1|\mathcal{E}\}$ in (\ref{equ_trdsa}) can be obtained from the fact that 
\begin{align}\label{equ_confirm}\notag
\Pr\{b> Z+1|\mathcal{E}\}=&1-\sum_{k=0}^{Z+1}\Pr\{b=k|\mathcal{E}\}\\
=&1-\sum_{k=0}^{Z+1}I^k(1-I)^{Z+1}{k+Z \choose k}.
\end{align}
Consequently, it is seen from (\ref{equ_trdsa}) that in order to derive $P_s^{(TR)}$, we only need to develop a closed-form expression for $P_L(Z+1-k)$.

In order to develop a closed-form expression for $P_L(Z+1-k)$, we consider a competition between an attacker's branch and an honest branch where the probabilities of the next block being mined and added in the attacker's branch and the honest branch are $I$ and $1-I$, respectively. The attacker wins this competition if the attacker's branch surpasses the honest branch before the honest branch grows by $l$ blocks. On the other hand, the attacker loses the competition when the honest branch grows by $l$ blocks and the attacker's branch doesn't surpass the honest branch. Let $Q(l,m,n)$ denote the probability that the attacker finally wins the competition in which at some time instant, the attacker's branch lags by $m$ blocks behind the honest branch and the honest branch has already grown by $n$ blocks, where $m=-1,0,...,$ and $n=0,1,...,l$. It is important to note that, if $m=Z+1-k$, $n=0$, and $l=L$, the probability that the attacker finally wins the competition, $Q(L,Z+1-k,0)$, is just the probability $P_L(Z+1-k)$. In other words, the probability $P_L(Z+1-k)$ just corresponds to a special case of the competition.
This motivates us to pursue a closed-form expression for $Q(l,m,n)$ to develop $P_L(Z+1-k)$, which is presented below.

Suppose that at a time instant, the attacker's branch lags by $m$ blocks behind the honest branch, and the honest branch has already grown by $n$ blocks during the competition. If the next block is mined and added in the attacker's branch, which happens with probability $I$, then the attacker's branch will lag by $(m-1)$ blocks behind the honest branch, and the length of the honest branch will remain unchanged. In contrast, if the next block is mined and added in the honest branch, which happens with probability $1-I$, the attacker's branch will lag by $(m+1)$ blocks behind the honest branch, and the honest branch will have grown by $(n+1)$ blocks. To this end, 
the recursive form of $Q(l,m,n)$ can be expressed as,  $m =0,1,...,$ and  $\forall n \in \{0,1,..., l-1\}$, 
\begin{align} \label{equ_2D}\notag
Q(l,m,n)=&I\times Q(l,m-1,n)\\
&+(1-I)\times Q(l,m+1,n+1).
\end{align}

It is worth mentioning that if the attacker's branch lags by $-1$ blocks behind the honest branch (i.e., the attacker's branch has surpassed the honest branch) before the honest branch grows by $l$ blocks during the competition (i.e., $n<l$), the attacker wins the competition. Therefore, we have the following boundary condition
\begin{equation} \label{boundary_condition_1}
	Q(l,-1,n)=1, \quad \forall n\in\{0,1,...,l-1\}.
\end{equation}
On the other hand, if the honest branch has grown by $l$ blocks during the competition while the attacker's branch is still not longer than the honest branch, then the attacker loses the competition.  Hence, we encounter another boundary condition
\begin{equation}\label{boundary_condition_2}
	Q(l,m,l)=0, \quad m =1,2,....
\end{equation}

\begin{figure}\centering 
	\includegraphics[width=0.5\textwidth]{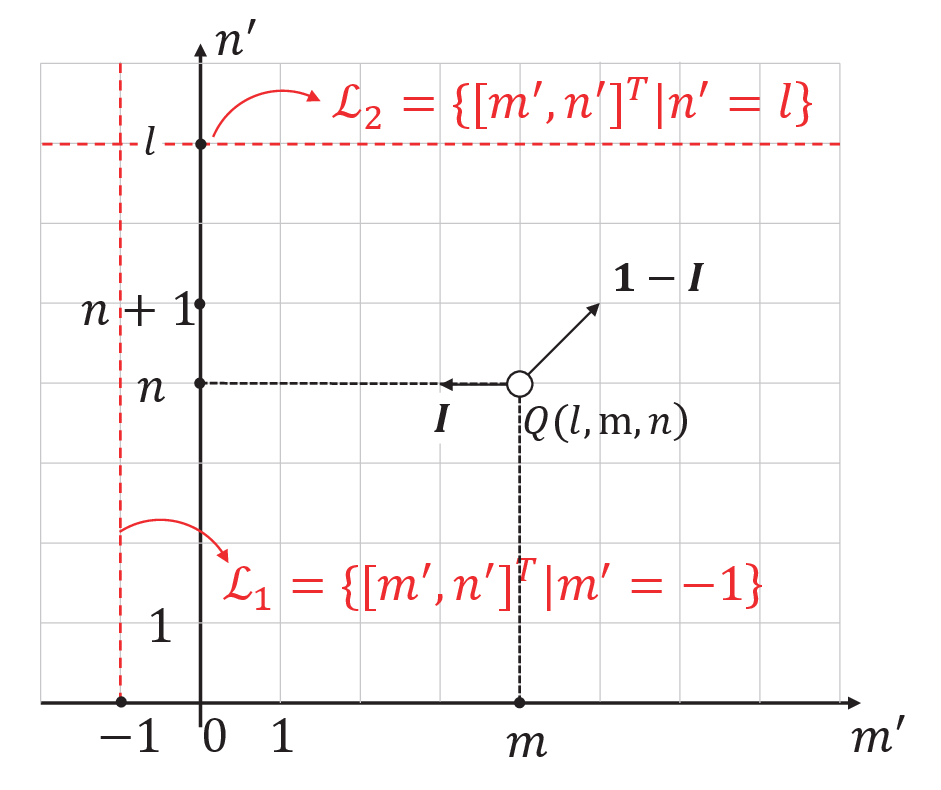}
	\caption{Two-dimensional random walk illustration.}
	\label{Fig_random_walk}
\end{figure}

By employing (\ref{equ_2D}), (\ref{boundary_condition_1}) and (\ref{boundary_condition_2}), the pursuit of the closed-form expression for $Q(l,m,n)$ can be cast as a two-sided boundary hitting problem for a two-dimensional random walk with two possible moving directions, which is illustrated in Fig. \ref{Fig_random_walk}. To be specific, let ${\boldsymbol s}_T \triangleq [m,n]^T + \sum_{t=1}^{T} \bdelta_t$ denote the point of the random walk in the two-dimensional space which starts at the initial point $[m,n]^T$ and moves by $T$ steps (i.e., the random walk moves by $T$ times which corresponds to the case where $T$ blocks have been mined in the blockchain). Hence, for any $t$, $\bdelta_t$ represents the movement of the random walk in the $t$-th step which is a random vector 
\begin{equation}
	{\bdelta _t} = \left\{ {\begin{array}{*{20}{c}}
{\bdelta^{(1)} \triangleq  {{[ - 1,0]}^T}}&{\text{with probability }I,}\\
{\bdelta^{(2)} \triangleq {{[1,1]}^T}}&{\text{with probability } 1 - I.}
\end{array}} \right.
\end{equation}
Let $\cL_1  \buildrel \Delta \over = \{ {[m',n']^T}\left| {m' =  - 1} \right.\}$ and ${\cL_2} \buildrel \Delta \over = \{ {[m',n']^T}\left| {n' = l} \right.\}$ define two boundary lines in the two-dimensional space, respectively. As such, $Q(l,m,n)$ is essentially the probability that the random walk hits $\cL_1$ before hitting $\cL_2$, which can be written as
\begin{equation} \label{P_m_n_closed_form}
	Q(l,m,n)\! =\! \sum_{T = 0}^\infty  \!{\Pr \left( {{\bs_T}\! \in\! {\cL_1},{\bs_{T'}}\! \notin {\cL_1} \cup {\cL_2},\;\forall T' < T} \right)}.
\end{equation}
For such a two-sided boundary hitting problem, the closed-form expression for $Q(l,m,n)$ for any given $l,m$, and $n$ is described in the following theorem.
\begin{theorem}\label{Theorem_DSA}
The probability $Q(l,m,n)$ that the attacker finally wins the competition in which at some time instant, the attacker's branch lags by $m$ blocks behind the honest branch and the honest branch has already grown by $n$ blocks can be expressed as, $m =-1,0,...$, and  $\forall n \in \{0,1,\cdots, l\}$,  
\begin{equation}\setlength{\arraycolsep}{1pt}
\begin{aligned} \label{P_m_n_theorem}
&Q(l,m,n)\\
&=\left\{\begin{array}{{c}{l}}\sum_{i=0}^{l-n-1}a_{i,m}(1-I)^iI^{m+1+i}&,\text{if }m\ge0,0\le n<l,\\
1&,\text{if }m=-1,0\!\le\!n\!<\!l,\\
0&,\text{if }m>0,n=l,\end{array}\right.
\end{aligned}
\end{equation}
where the coefficient $a_{i,m}$ is defined as
\begin{equation}\label{a_i_m_theorem}
\setlength{\arraycolsep}{1pt}
{a_{i,m}} = \left\{ {\begin{array}{*{5}{cl}}
	{1,}&{{\rm{if }} \; i = 0,}\\
	{1 + m,}&{{\rm{if }}\; i = 1,}\\
	{{C_i},}&{{\rm{if }}\; m = 0,}\\
	{{C_{i + 1}},}&{{\rm{if }}\; m = 1,}\\
	\begin{array}{l}
	{C_{i + 1}} + \sum\limits_{{j_1} = 3}^{m + 1} {{C_i}} + \sum\limits_{{j_1} = 3}^{m + 1} {\sum\limits_{{j_2} = 3}^{{j_1} + 1} {{C_{i - 1}}} }\\
 \quad    +  \cdots  
	   + \sum\limits_{{j_1} = 3}^{m + 1} {\sum\limits_{{j_2} = 3}^{{j_1} + 1}  \cdots  } \sum\limits_{{j_{i - 2}} = 3}^{{j_{i - 3}} + 1} {{C_3}} \\
	\quad   + \sum\limits_{{j_1} = 3}^{m + 1} {\sum\limits_{{j_2} = 3}^{{j_1} + 1}  \cdots  } \sum\limits_{{j_{i - 1}} = 3}^{{j_{i - 2}} + 1} {(1 + {j_{i-1}})} ,
	\end{array}&\ \begin{aligned}&{\rm{if }} \; i \! > \!1\\&\rm{and}\;m > 1,\end{aligned}
	\end{array}} \right.
\end{equation} 	
and the constant $C_i$ is the $i$-th Catalan number which is given by
\begin{equation}
	C_i \buildrel \Delta \over =  \frac{1}{i+1}{2i\choose i} = \frac{{(2i)!}}{{(i + 1)!i!}}.
\end{equation}
\end{theorem}

\begin{proof}[Proof:\nopunct]
From (\ref{equ_2D}) and (\ref{boundary_condition_1}), it is easy to see that $Q(l,m,n)=1$ if $m=-1$ and $Q(l,m,n)=0$ if $n=l$. Therefore, we just need to prove the theorem for the case that $m\ge0$ and $0\le n<l$.

	As illustrated in Fig. \ref{Fig_random_walk}, if the random walk can hit the boundary line $\cL_1$ before hitting the boundary line $\cL_2$, all the possible cases are that $\forall i=0,1,...,l-n-1$, the random walk moves along $\bdelta^{(1)}$ by $(m+1+i)$ steps and moves along $\bdelta^{(2)}$ by $i$ steps.
Thus, $Q(l,m,n)$ can be written in the general form
\begin{equation} \label{equ_genForm}
Q(l,m,n)=\sum_{i=0}^{l-n-1}a_{i,m}(1-I)^iI^{m+1+i},
\end{equation}
where $a_{i,m}$ is a constant which denotes the number of all the paths that the random walk hits the boundary $\cL_1$ in the $(m+2i+1)$-th step by moving $(m+i+1)$ steps along $\bdelta^{(1)}$  and $i$ steps along $\bdelta^{(2)}$ and the random walk never hits $\cL_1$ before the $(m+2i+1)$-th step.  It is obvious that 
\begin{equation} \label{a_0_m}
	a_{0,m}=1 \;\; \text{for  } m \ge 0,
\end{equation}
and hence
\begin{equation} \label{P_m_n_alt}
	Q(l,m,n)=I^{(m+1)} + \sum_{i=1}^{l-n-1}a_{i,m}(1-I)^iI^{m+1+i}.
\end{equation}
By substituting (\ref{equ_genForm}) and (\ref{P_m_n_alt}) into (\ref{equ_2D}), we can obtain
\begin{align} \notag
&I^{m+1}+\sum_{i=1}^{l-n-1}a_{i,m}(1-I)^iI^{m+1+i}\\ \notag
&=I\times\left(\sum_{i=0}^{l-n-1}a_{i,m-1}(1-I)^iI^{m+i}\right) \\ \notag
&\qquad  +(1-I)\times \left(\sum_{i=0}^{l-n-2}a_{i,m+1}(1-I)^iI^{m+2+i}\right)\\  \notag
&=\sum_{i=0}^{l-n-1}a_{i,m-1}(1-I)^iI^{m+1+i}\\\notag
&\qquad+\sum_{i=0}^{l-n-2}a_{i,m+1}(1-I)^{i+1}I^{m+2+i}\\  \notag
&=I^{m+1}+\sum_{i=1}^{l-n-1}a_{i,m-1}(1-I)^iI^{m+1+i}\\  \notag
&\qquad  +\sum_{i=1}^{l-n-1}a_{i-1,m+1}(1-I)^{i}I^{m+1+i}\\ \label{I_poly}
&=I^{m+1}+\sum_{i=1}^{l-n-1}[a_{i,m-1}+a_{i-1,m+1}](1-I)^iI^{m+1+i}.
\end{align}
Since (\ref{I_poly}) holds for any $I$, by the fundamental theorem of algebra, we know 
\begin{equation} \label{equ_Recur}
a_{i,m}=a_{i,m-1}+a_{i-1,m+1}, \;\; \forall i>0 \text{ and }m\ge0.
\end{equation}
Since $a_{0,m}=1$ for any $m \ge 0$, we can get \begin{equation} \label{a_1_m}
	a_{1,m}=a_{1,m-1}+1 = a_{1,0} + m ,\;\; \forall m\ge 0
\end{equation}
from (\ref{equ_Recur}) by choosing $i=1$.
Furthermore, by setting $m=0$ in (\ref{equ_2D}), we have 
\begin{equation}
	Q(l,0,n)=I+(1-I)\times Q(l,1,n+1),
\end{equation}
which yields 
\begin{align}\label{I_poly_2} \notag
&I+\sum_{i=1}^{l-n-1}a_{i,0}(1-I)^{i}I^{1+i}\\
&=I+\sum_{i=1}^{l-n-1}a_{i-1,1}(1-I)^{i}I^{1+i},
\end{align}
by employing (\ref{equ_genForm}).
By the fundamental theorem of algebra, we know from (\ref{I_poly_2}) that
\begin{equation} \label{a_recursive}
a_{i,0}=a_{i-1,1}, \;\; \forall i=1,...,l-n-1.
\end{equation}
From (\ref{a_0_m}) and (\ref{a_recursive}), we can obtain
\begin{equation}
	a_{1,0}=a_{0,1}=1,
\end{equation}  
and therefore, 
\begin{equation} \label{a_1_m_final}
	a_{1,m}=1+m \;\; \forall m\ge 0,
\end{equation}
by employing (\ref{a_1_m}).
From the recursive equation in (\ref{equ_Recur}), we can obtain that for any $i =1,2,...,l-n-1$,
\begin{equation} \label{a_i_m_1}
	a_{i,m}=a_{i,1}+\sum_{j=3}^{m+1}a_{i-1,j}, \;\forall m>1.
\end{equation}
Moreover, from (\ref{a_recursive}), we know
\begin{equation} \label{a_i_1_2}
	a_{i,1}=a_{i+1,0}, \;\; \forall i=0,...,l-n-2,
\end{equation} 
which implies that for all $i=1,...,l-n-2$ and $m > 1$,
\begin{align} \notag
a_{i,m}&=a_{i+1,0}+\sum_{j_1=3}^{m+1}a_{i-1,j_1}\\  \notag
&=a_{i+1,0}+\sum_{j_1=3}^{m+1}a_{i,0}+\sum_{j_1=3}^{m+1}\sum_{j_2=3}^{j_1+1}a_{i-1,0}+\cdots\\ \notag
&\quad+\sum_{j_1=3}^{m+1}\sum_{j_{2}=3}^{j_{2}+1}\cdots\sum_{j_{i-2}=3}^{j_{i-3}+1}a_{3,0}+\sum_{j_1=3}^{m+1}\sum_{j_{2}=3}^{j_{2}+1}\cdots\sum_{j_{i-1}=3}^{j_{i-2}+1}a_{1,j_{i-1}} \\ \notag
&=a_{i+1,0}+\sum_{j_1=3}^{m+1}a_{i,0}+\sum_{j_1=3}^{m+1}\sum_{j_2=3}^{j_1+1}a_{i-1,0}+\cdots\\ \notag
&\quad+\sum_{j_1=3}^{m+1}\sum_{j_{2}=3}^{j_{2}+1}\cdots\sum_{j_{i-2}=3}^{j_{i-3}+1}a_{3,0} \\  \label{a_i_m_final_1}
& \quad +\sum_{j_1=3}^{m+1}\sum_{j_{2}=3}^{j_{2}+1}\cdots\sum_{j_{i-1}=3}^{j_{i-2}+1}(1+j_{i-1}),
\end{align}
by employing (\ref{a_1_m_final}) and (\ref{a_i_m_1}).

Next, we need to determine $a_{i,m}$ when $i=l-n-1$.
Suppose that there is another boundary line ${\cL_3} \buildrel \Delta \over = \{ {[m',n']^T}\left| {n' = l+1} \right.\}$ and $Q(l+1,m,n)$ is the probability that the random walk hits the boundary line $\cL_1$ before hitting the boundary line $\cL_3$. From (\ref{equ_genForm}), we know 
$$Q(l+1,m,n)=\sum_{i=0}^{l-n}a_{i,m}(1-I)^iI^{m+1+i}.$$
Similar to (\ref{a_i_1_2}), we can get 
\begin{equation} \label{a_prime_rec}
	a_{i,1}=a_{i+1,0}, \;\; \forall i=0,...,l-n-1.
\end{equation}
From (\ref{a_prime_rec}), we know
\begin{equation}
a_{l-n-1,1}=a_{l-n,0}.
\end{equation}
Similar to (\ref{a_i_m_final_1}), we can obtain that for all $m > 1$,
\begin{align} \notag
a_{l-n-1,m} & =a_{l-n,0}+\sum_{j_1=3}^{m+1}a_{l-n-2,j_1}\\ \notag
&=a_{l-n,0}+\sum_{j_1=3}^{m+1}a_{l-n-1,0}+\sum_{j_1=3}^{m+1}\sum_{j_2=3}^{j_1+1}a_{l-n-2,0}\\ \notag
&\qquad+\cdots+\sum_{j_1=3}^{m+1}\sum_{j_{2}=3}^{j_{2}+1}\cdots\sum_{j_{l-n-3}=3}^{j_{l-n-4}+1}a_{3,0} \\ \label{a_i_m_final_2}
& \qquad +\sum_{j_1=3}^{m+1}\sum_{j_{2}=3}^{j_{2}+1}\cdots\sum_{j_{l-n-2}=3}^{j_{l-n-3}+1}(1+j_{l-n-2}).
\end{align}

It is seen from (\ref{a_0_m}), (\ref{a_i_m_final_1}), and (\ref{a_i_m_final_2}) that if we can determine the value of $a_{i,0}$ for all $i=3,4,...,l-n$, then we can determine the value of $a_{i,m}$ for all $i=0,1,...,l-n-1$ and $\forall m \ge 0$, and hence, we can obtain the closed-form expression for $Q(l,m,n)$ by using (\ref{equ_genForm}). In what follows, we will derive the closed-form expression for $a_{i,0}$ for all $i=3,4,...,l-n$. 



Note that  $a_{i,0}$ denotes the number of all the paths that the random walk hits the boundary line $\cL_1$ in the $(2i+1)$-th step by moving $(i+1)$ steps along $\bdelta^{(1)}$  and $i$ steps along $\bdelta^{(2)}$ and the random walk never hits $\cL_1$ before the $(2i+1)$-th step. We know that for any such path in the two-dimensional space, it starts from the point $[0,n]^T$ and arrives at the point $[0,n+i]^T$ of the $(2i)$-th step, and moreover, the first $2i$ steps of the path must stay in the half space $\cS \triangleq  \{ {[m',n']^T}\left| {m' \ge 0} \right.\}$.
In light of this, $a_{i,0}$ is identical to the number of all the lattice paths from the point $[0,0]^T$ to the point $[i,i]^T$ which consist of $i$  steps along the vector $[0,1]^T$ and $i$  steps along the vector $[1,0]^T$ and never rise above the diagonal line in the $i$-by-$i$ grid. 
Such paths are referred to as the Dyck Paths, and the number of such paths is known as the $i$-th Catalan number $C_i$ \cite{stanley2015catalan}. Therefore, we know
\begin{align} \label{a_prime_c_i} \notag
	a_{i,0}&=C_i \\
&\triangleq  \frac{1}{i+1}{2i\choose i} = \frac{{(2i)!}}{{(i + 1)!i!}}, \;\; \forall i=3,...,l-n-1.
\end{align}
By employing (\ref{boundary_condition_1}), (\ref{boundary_condition_2}), (\ref{equ_genForm}), (\ref{a_0_m}), (\ref{a_1_m_final}), (\ref{a_i_m_final_1}), (\ref{a_prime_rec}),  (\ref{a_i_m_final_2}) and (\ref{a_prime_c_i}), we can obtain $m =-1,0,...$ and  $\forall n \in \{0,1,\cdots, l\}$, 
\begin{equation}\setlength{\arraycolsep}{1pt}
\begin{aligned} 
&Q(l,m,n)\\
&=\left\{\begin{array}{{c}{l}}\sum_{i=0}^{l-n-1}a_{i,m}(1-I)^iI^{m+1+i}&,\text{if }m\ge0,0\le n<l,\\
1&,\text{if }m=-1,0\!\le\!n\!<\!l,\\
0&,\text{if }m>0,n=l,\end{array}\right.
\end{aligned}
\end{equation}
where the coefficient $a_{i,m}$ can be expressed as

\begin{equation}
\setlength{\arraycolsep}{1pt}
{a_{i,m}} = \left\{ {\begin{array}{*{5}{c}}
	{1,}&{{\rm{if }} \; i = 0,}\\
	{1 + m,}&{{\rm{if }}\; i = 1,}\\
	{{C_i},}&{{\rm{if }}\; m = 0,}\\
	{{C_{i + 1}},}&{{\rm{if }}\; m = 1,}\\
	\begin{array}{l}
	{C_{i + 1}} + \sum\limits_{{j_1} = 3}^{m + 1} {{C_i}} + \sum\limits_{{j_1} = 3}^{m + 1} {\sum\limits_{{j_2} = 3}^{{j_1} + 1} {{C_{i - 1}}} }\\
 \quad   +  \cdots  
	   + \sum\limits_{{j_1} = 3}^{m + 1} {\sum\limits_{{j_2} = 3}^{{j_1} + 1}  \cdots  } \sum\limits_{{j_{i - 2}} = 3}^{{j_{i - 3}} + 1} {{C_3}} \\
	\quad   + \sum\limits_{{j_1} = 3}^{m + 1} {\sum\limits_{{j_2} = 3}^{{j_1} + 1}  \cdots  } \sum\limits_{{j_{i - 1}} = 3}^{{j_{i - 2}} + 1} {(1 + {j_{i-1}})} ,
	\end{array}&\ \begin{aligned}&{\rm{if }} \; i  > 1\\&\rm{and}\;m > 1,\end{aligned}
	\end{array}} \right.
\end{equation} 	
which completes the proof.

\end{proof}

From (\ref{equ_trdsa}), (\ref{equ_bio}), (\ref{equ_confirm}), and (\ref{P_m_n_theorem}), we can obtain the success probability of a TR-DSA,
\begin{align}\label{equ_trdsa_1}\notag
P_s^{(TR)}=&\Pr\{b> Z+1|\mathcal{E}\}\\\notag
&+\sum_{k=0}^{Z+1}\Pr\{b=k|\mathcal{E}\}P_L(Z+1-k)\\\notag
=&1-\sum_{k=0}^{Z+1}\Pr\{b=k|\mathcal{E}\}\\\notag
&+\sum_{k=0}^{Z+1}\Pr\{b=k|\mathcal{E}\}Q(L,Z+1-k,0)\\ \notag
=&1-\sum_{k=0}^{Z+1}\Bigg\{{k+Z\choose k} I^k(1-I)^{Z+1}\\
&\times\left[1-\sum_{i=0}^{L-1}a_{i,Z+1-k}(1-I)^iI^{Z+2-k+i}\right]\Bigg\},
\end{align}
where $a_{i,Z+1-k}$ is defined in (\ref{a_i_m_theorem}) for any $Z$ and $k$.

As previously mentioned, the TR-DSA model is well suited to describe DSAs on blockchain applications that require tasks to be completed within finite timeframes. With the derived success probability of TR-DSA from (\ref{equ_trdsa_1}), these applications can measure their vulnerability to double-spending attacks. In particular, they can calculate the required number of blocks $Z$ for transaction confirmation, which can ensure that the success probability of a TR-DSA is smaller than a prescribed threshold, and hence substantially reduce the risk of a TR-DSA in their applications. Additionally, the TR-DSA model is appropriate for describing DSAs for scenarios where a practical attacker cannot indefinitely invest its computational resources in conducting a single DSA. In light of this, the developed probability in (\ref{equ_trdsa_1}) can help potential practical attackers to evaluate the feasibility and expected reward of launching a DSA.
\subsection{Comparison of the Success Probability of A TR-DSA and the Success Probability of A TU-DSA}\label{Section_vs}
In this subsection, we compare the success probability of a TR-DSA with that of a TU-DSA.
We start by revisiting the success probability of a TU-DSA.
In the TU-DSA model where an attacker can indefinitely mine its branch to surpass the honest branch, the pursuit of the probability that the attacker's branch will surpass the honest branch when the attacker's branch lags by $m$ blocks behind the honest branch can be characterized as an infinite Gambler's Ruin problem \cite{nakamoto2008bitcoin,rosenfeld2014analysis,grunspan2018double,ozisik2017explanation,zaghloul2020bitcoin}, and this probability, denoted as $P(m)$, is
\begin{align}\label{equ_tudsa_catch}
P(m)=\left\{{\begin{array}{*{5}{c}} 
\left(\frac{I}{1-I}\right)^{m+1}&,\text{ if }\;0\le I<0.5,\\
1&,\text{ if }\;0.5\le I\le 1,
\end{array}}\right.
\end{align}
 and the success probability of a TU-DSA, $P_s^{(TU)}$, is \cite{nakamoto2008bitcoin,rosenfeld2014analysis,grunspan2018double,ozisik2017explanation,zaghloul2020bitcoin},
\begin{align}\label{equ_tudsa}\notag
P_s^{(TU)}=&1-\sum_{k=0}^{Z+1}\Pr\{b=k|\mathcal{E}\}\\\notag
&+\sum_{k=0}^{Z+1}\Pr\{b=k|\mathcal{E}\}P(Z+1-k)\\ \notag
=&1-\sum_{k=0}^{Z+1}\Bigg\{{k+Z\choose k}I^k(1-I)^{Z+1}\\ 
&\times\left[1-\left(\frac{I}{1-I}\right)^{Z+2-k}\right]\Bigg\}.
\end{align}

It is seen from equations (\ref{equ_trdsa_1}) and (\ref{equ_tudsa}) that  the essence of comparison between $P_s^{(TR)}$ and $P_s^{(TU)}$ is the relationship between $Q(L,m,n)$ and $P(m)$. In what follows, we will concentrate on contrasting $Q(L,m,n)$ and $P(m)$.

Intuitively speaking, if an attacker launches a TR-DSA but with $L\rightarrow\infty$, this TR-DSA should degenerate to a TU-DSA because as $L\rightarrow\infty$, the attacker essentially can indefinitely mine its branch. Hence, the probability $Q(l,m,n)$ is expected to converge to the probability $P(m)$ as $l\to\infty$. Furthermore, as $L$ increases, the attacker is anticipated to gain a greater opportunity to mine its branch to be longer than the honest branch. This implies that $Q(l,m,n)$ is expected to increase as $l$ increases. 
The following theorem provides a formal summary and rigorous proof of these intuitive conjectures which relate $Q(l,m,n)$ and $P(m)$.


\begin{theorem}\label{Theorem_asymptotic}
If $0\le n<l$, then $Q(l,m,n)$ strictly increases as $l$ increases when $0<I<1$, and is constant when $I=0$. Moreover, when $0\le n<l$, as $l \to \infty$, the probability $Q(l,m,n)$ in (\ref{P_m_n_theorem}) converges to
\begin{equation} \label{P_m_n_asymptotic}
\lim_{l\to\infty}Q(l,m,n)=\left\{{\begin{array}{*{5}{c}} 
\left(\frac{I}{1-I}\right)^{m+1}&,\text{ if }\;0\le I<0.5,\\
1&,\text{ if }\;0.5\le I< 1,
\end{array}}\right.
\end{equation}
which equals the probability $P(m)$ in (\ref{equ_tudsa_catch}).
\end{theorem}
\begin{proof}[Proof:\nopunct]
If $0\le n<l$, then from (\ref{P_m_n_theorem}), we have
\begin{align}\notag
Q(l+1,m,n)&=\sum_{i=0}^{l-n}a_{i,m}(1-I)^iI^{m+1+i}\\
&=Q(l,m,n)+a_{{l-n},m}(1-I)^{l-n}I^{m+1+l-n}.
\end{align}
From (\ref{a_i_m_theorem}), we know that $a_{i,m}\ge1$,$\forall i,m\ge0$, and hence, we have $a_{{l-n},m}(1-I)^{l-n}I^{m+1+l-n}>0$ when $0<I<1$ and $a_{{l-n},m}(1-I)^{l-n}I^{m+1+l-n}=0$ when $I=0$. This implies that 
\begin{align}\label{P_increasing}
Q(l+1,m,n)>Q(l,m,n)&, \;\; \text{when }  0<I<1,\\ \label{P_constant}
Q(l+1,m,n)=Q(l,m,n)&, \;\; \text{when } I=0 .
\end{align}
From (\ref{P_increasing}) and (\ref{P_constant}), we know that when $0<I<1$,  $Q(l,m,n)$ is a strictly increasing function of $l$,  and when $I=0$, $Q(l,m,n)$ remains constant as $l$ varies. 
By the definition of $Q(l,m,n)$, we know that $Q(l,m,n)\le1$, $\forall l$. Thus, $Q(l,m,n)$ is strictly increasing and bounded from above when $0<I<1$, and therefore, we know $Q(l,m,n)$ converges as $l$  increases when $0<I<1$ by the monotone convergence theorem \cite{wheeden1977measure}. When $I=0$, $Q(l,m,n)$ doesn't change as $l$ changes, and hence, $Q(l,m,n)$ also converges. 
Therefore, $Q(l,m,n)$ converges as $l$ increases, which implies
\begin{equation}\label{p_infty}
\lim_{l\to\infty}Q(l,m,n)=\lim_{l\to\infty}Q(l+1,m,n).\end{equation}
Note that 
\begin{align}\label{p_l+1} \notag
Q&(l+1,m+1,n+1)\\ \notag
&=\sum_{i=0}^{(l+1)-(n+1)-1}a_{i,m+1}(1-I)^iI^{(m+1)+1+i}\\ \notag
&=\sum_{i=0}^{(l)-(n)-1}a_{i,m+1}(1-I)^iI^{(m+1)+1+i}\\
&=Q(l,m+1,n).
\end{align}
From (\ref{p_infty}) and (\ref{p_l+1}), we can get  
\begin{equation}\label{p_n-1}
\lim_{l\to\infty}Q(l,m+1,n+1)=\lim_{l\to\infty}Q(l,m+1,n), \;\; \text{if }n\ge 0.
\end{equation}
From (\ref{equ_2D}), we can get 
\begin{align}\label{equ_2d_infty}\notag
\lim_{l\to\infty}Q(l,m,n)=&\ I\times \lim_{L\to\infty}Q(l,m-1,n)\\
&+(1-I)\times \lim_{l\to\infty}Q(l,m+1,n+1).
\end{align}
Substituting (\ref{p_n-1}) into (\ref{equ_2d_infty}), we have
\begin{align}\label{equ_1d_infty}\notag
\lim_{l\to\infty}Q(l,m,n)=&\ I\times \lim_{l\to\infty}Q(l,m-1,n)\\
&+(1-I)\times \lim_{l\to\infty}Q(l,m+1,n).
\end{align}
For simplicity, let $g(m+1)\triangleq \lim_{l\to\infty}Q(l,m,n)$, $\forall m\ge -1$,  then from (\ref{equ_1d_infty}), we can obtain
\begin{equation}\label{equ_1d_g}
g(m+1)=I\times g(m)+(1-I)\times g(m+2).
\end{equation}
From (\ref{boundary_condition_1}), we have $\lim_{l\to\infty}Q(l,-1,n)=1$ which implies 
\begin{equation}\label{g_0}
g(0)=\lim_{l\to\infty}Q(l,-1,n)=1.
\end{equation}
From (\ref{equ_1d_g}), we have
\begin{equation}\label{equ_1d_grecursive}
g(m+2)-g(m+1)=\frac{I}{1-I}\left[g(m+1)-g(m)\right],
\end{equation}
which yields
 \begin{equation} \label{equ_1d_grecursive1}
g(m+1)-g(m)=\left(\frac{I}{1-I}\right)^m\left[g(1)-g(0)\right].
\end{equation}
Moreover, from the recursive equation in (\ref{equ_1d_grecursive1}), we can obtain
 \begin{align} \label{equ_1d_grecursive2}
g(m+1)-g(1)=\left[g(1)-g(0)\right]\sum_{i=1}^m\left(\frac{I}{1-I}\right)^m.
\end{align}
Next, we will determine $g(1)$. From (\ref{P_m_n_theorem}) and (\ref{a_i_m_theorem}), we know that for any $n$,
\begin{align}\label{g_1}\notag
g(1)=&\lim_{l\to\infty}Q(l,0,n)\\ \notag
=&\lim_{l\to\infty}\sum_{i=0}^{l-n-1}C_{i}(1-I)^iI^{1+i} \\ 
= &\sum_{i=0}^{\infty}C_{i}(1-I)^iI^{1+i},
\end{align}
where $C_i$ is the $i$-th Catalan number. 
Note that the generating function $c(x)$ for Catalan numbers $\{C_i\}$ is defined as \cite{stanley2015catalan}
\begin{align}\label{generating_catalan}
c(x)\triangleq\sum_{i=0}^{\infty}C_{i}x^{i}=\frac{1-\sqrt{1-4x}}{2x}.
\end{align}
By setting $x=I(1-I)$ for some $I\in (0,1)$ in (\ref{generating_catalan}), we can obtain 
\begin{align}\label{generating_I} \notag
\sum_{i=0}^{\infty}C_{i}\left[I(1-I)\right]^{i}=&\frac{1-\sqrt{1-4I(1-I)}}{2I(1-I)}\\ \notag
=&\frac{1-\sqrt{(1-2I)^2}}{2I(1-I)}\\ 
=&\left\{\begin{array}{cc}\frac{1}{1-I}&,0<I<0.5,\\\frac{1}{I}&,0.5\le I<1.\end{array}\right.
\end{align}
Note that if $I=0$, then we know from (\ref{g_1}) that $g(1)=0$. Moreover, since $I\times\sum_{i=0}^{\infty}C_{i}[I(1-I)]^{i}=\sum_{i=0}^{\infty}C_{i}(1-I)^iI^{1+i}=g(1)$, 
we can obtain from (\ref{g_1}) and (\ref{generating_I}) that
\begin{align}\label{equ_g1} 
g(1)=\left\{\begin{array}{cc}\frac{I}{1-I}&,0\le I<0.5,\\1&,0.5\le I< 1.\end{array}\right.
\end{align}
By employing (\ref{g_0}), (\ref{equ_1d_grecursive2}) and (\ref{equ_g1}), we can obtain
\begin{equation}
\lim_{l\to\infty}Q(l,m,n)\!=\!\left\{\begin{array}{*{5}{c}}(\frac{I}{1-I})^{m+1}&,\text{if }0\le  I<0.5,\\1&,\text{if }0.5\le I< 1,\end{array}\right.\!=\!P(m),
\end{equation}
which completes the proof.
\end{proof}
As indicated by Theorem \ref{Theorem_asymptotic}, $Q(l,m,n)$ is a non-decreasing function of $l$, and it converges to the probability $P(m)$ as $l$ increases. This indicates that $P(m)$ serves as an upper bound for $Q(l,m,n)$. Additionally, it is seen from Theorem \ref{Theorem_asymptotic} that for a finite $L$, $Q(L,m,n)\le\lim_{l\to\infty}Q(l,m,n)$ with equality holding only when $I=0$. This implies that $Q(L,m,n)$ remains strictly smaller than 1 when $0\le I<1$ and $L<\infty$. 
Moreover, by employing the results of Theorem \ref{Theorem_asymptotic}, we can compare the success probability of a TR-DSA and that of a TU-DSA, which is formally summarized in the following corollary.
\begin{corollary}\label{Corollary_1}
The success probability of a TR-DSA strictly increases as $L$ increases when $0<I<1$. In addition, for any $L$, the success probability of a TR-DSA is no greater than the success probability of a TU-DSA. Particularly, the success probability of a TR-DSA is strictly smaller than the success probability of a TU-DSA when $0< I<1$. Moreover, the success probability of a TR-DSA is strictly smaller than 1 when $0\le I<1$.
\end{corollary}
\begin{proof}
It is seen from (\ref{equ_trdsa_1}) that the success probability of a TR-DSA increases as $Q(L,Z+1-k,0)$ increases. Since we know from Theorem \ref{Theorem_asymptotic} that $Q(l,m,n)$ strictly increases as $l$ increases when $0<I<1$, we can conclude that the success probability of a TR-DSA strictly increases as $L$ increases when $0<I<1$.

From  Theorem \ref{Theorem_asymptotic}, we know that for a finite $L$, $Q(L,m,n)\le \lim_{l\to\infty}Q(l,m,n)=P(m)$. From (\ref{equ_trdsa_1}), (\ref{equ_tudsa}), and Theorem \ref{Theorem_asymptotic}, we can obtain,
\begin{align}  \notag
P_s^{(TR)}=&1-\sum_{k=0}^{Z+1}\Pr\{b=k|\mathcal{E}\}\\ \label{pr_trdsa}
&+\sum_{k=0}^{Z+1}\Pr\{b=k|\mathcal{E}\}Q(L,Z+1-k,0)\\ \notag
\le&1-\sum_{k=0}^{Z+1}\Pr\{b=k|\mathcal{E}\}  \\ 
&+\sum_{k=0}^{Z+1}\Pr\{b=k|\mathcal{E}\}\lim_{l\to\infty}Q(l,Z+1-k,0)\label{pr_tudsa}\\ \notag
=&1-\sum_{k=0}^{Z+1}\Pr\{b=k|\mathcal{E}\}\\ \notag
&+\sum_{k=0}^{Z+1}\Pr\{b=k|\mathcal{E}\}P(Z+1-k)\\  \notag
=&P_s^{(TU)},
\end{align}
where the inequality in (\ref{pr_tudsa}) becomes equality only when $I=0$. Moreover, we know from (\ref{pr_trdsa}) that only when $Q(L,Z+1-k,0)=1,\forall k\in\{0,1,...,Z+1\}$, the success probability of a TR-DSA is one. But it is seen from  Theorem \ref{Theorem_asymptotic} that $Q(L,m,n)$ is strictly smaller than one when $0\le I<1$ and $L<\infty$. This implies that the success probability of a TR-DSA is strictly smaller than one when $0\le I<1$. The proof is now completed.
\end{proof}
From (\ref{equ_tudsa_catch}) and (\ref{equ_tudsa}), we know that the success probability of a TU-DSA is one when $0.5\le I<1$, which implies that if an attacker gains the majority of the hash rate of a blockchain network and can indefinitely mine its branch, it has the capability to successfully falsify any transactions stored in the main chain of the blockchain. However, Corollary \ref{Corollary_1} reveals that the success probability of a TR-DSA is strictly smaller than one even when $0.5\le I<1$. This means that  for an attacker with limited computational resources, even if it amasses the majority of the network's hash rate, the risk of failure in launching a TR-DSA persists. Additionally, since the success probability of a TR-DSA isn't higher than the success probability of a TU-DSA, it underscores that blockchain applications with finite timeframes for task completion are inherently less vulnerable to double-spending attacks than those allowing unlimited timeframes.

\section{Numerical Results}
\label{Section_simulation}

In this section, we first numerically study the probability $Q(l,m,n)$ that an attacker's branch successfully catches up with and surpasses the honest branch for different parameters $m$, $l$, $n$, and $I$, respectively.  In particular, we employ Monte Carlo simulations to corroborate the theories developed in this paper. The number of Monte Carlo runs is $10^4$ in all simulation experiments, and the Monte Carlo simulation results are specified by the legend label ``Simulation'' in the figures. In addition, we numerically compare the success probability of a TR-DSA with that of a TU-DSA.
\begin{figure}[H]\centering   
	\includegraphics[width=0.5\textwidth]{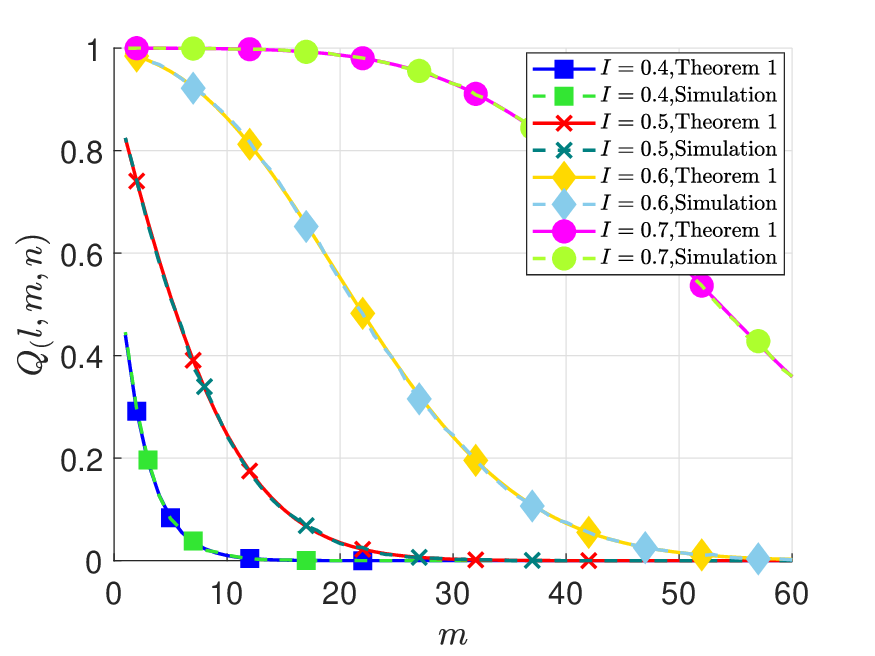}    
	\caption{The relationship between $Q(l,m,n)$ and $m$ for different $I$. }
	\label{fig_PI}
\end{figure}

\begin{figure}[H]\centering   
	\includegraphics[width=0.5\textwidth]{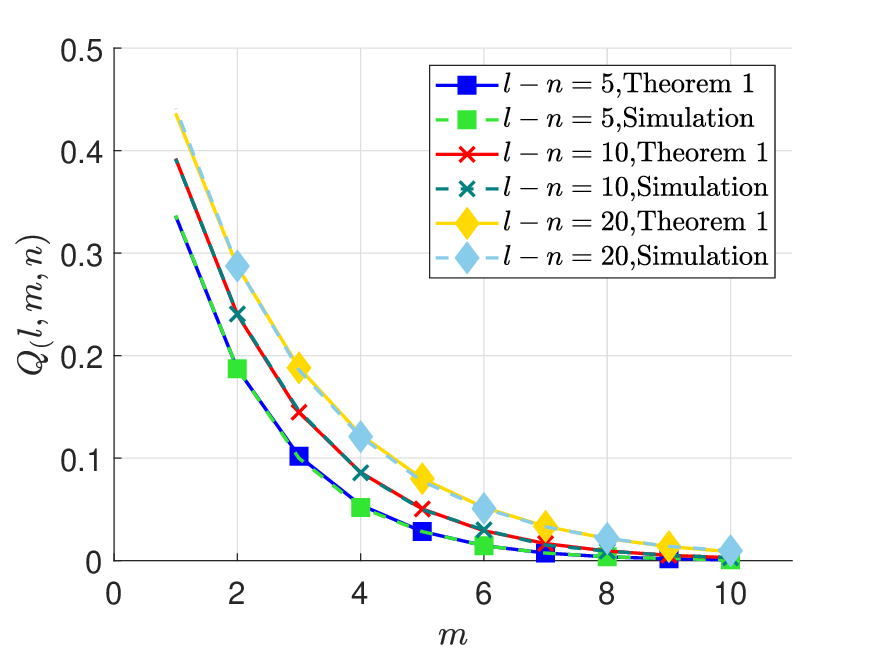}    
	\caption{The relationship between $Q(l,m,n)$ and $m$ for different $(l-n)$.}
	\label{fig_PL}
\end{figure}

As $m$ varies from $1$ to $60$, Fig. \ref{fig_PI} depicts $Q(l,m,n)$ for different $I$, where $n=0$ and $l=40$. The numerical results yielded from Theorem \ref{Theorem_DSA} and Monte Carlo simulations are specified by solid and dashed curves in Fig. \ref{fig_PI}, respectively, which clearly agree with each other. Hence, the numerical results in Fig. \ref{fig_PI} corroborate Theorem \ref{Theorem_DSA}.  It is seen from Fig. \ref{fig_PI} that $Q(l,m,n)$ increases as $m$ decreases, and moreover, for a given $m$, $Q(l,m,n)$ increases as $I$ increases.  This is because when $m$ decreases, the gap between the number of blocks in the attacker's branch and that in the honest branch shrinks, 
and therefore, it is easier for the attacker's branch to surpass the honest branch and hence win the competition. Moreover, when $I$ increases, the probability that the next block mined in the blockchain is generated by the attacker increases, and hence, it is also easier for the attacker's branch to win the competition.
\begin{figure}[H]\centering   
	\includegraphics[width=0.5\textwidth]{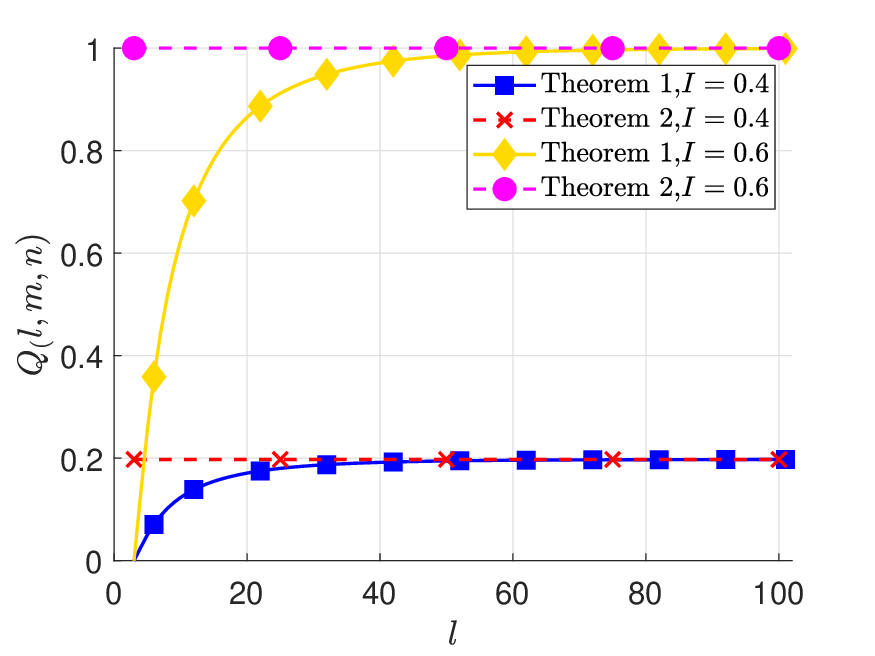}    
	\caption{The comparison between $Q(l,m,n)$  obtained from Theorem \ref{Theorem_DSA} and Theorem \ref{Theorem_asymptotic}.}
	\label{fig_vs}
\end{figure}
Note that from (\ref{P_m_n_theorem}), we know that $Q(l,m,n)$ varies with $(l-n)$ which is the number of remaining blocks that the honest miners have to generate to win the competition. Next, we numerically investigate how the value of $(l-n)$ impacts $Q(l,m,n)$.  As $m$ varies from $1$ to $10$, Fig. \ref{fig_PL} depicts $Q(l,m,n)$ for different $(l-n)$, where $I=0.4$. It is seen from Fig. \ref{fig_PL} that the numerical results yielded from Theorem 1 agree with those from Monte Carlo simulations. Moreover, for a given $m$, $Q(l,m,n)$ increases as $(l-n)$ increases. This can be explained by the fact that when $(l-n)$ increases, it takes more time for the honest branch to grow by the designated number of blocks, and hence, the attacker's branch has a better chance to win the competition.



 We further numerically corroborate Theorem \ref{Theorem_asymptotic}. As $l$ increases, Fig. \ref{fig_vs} depicts $Q(l,m,n)$ obtained from (\ref{P_m_n_theorem}) and (\ref{P_m_n_asymptotic}) for different $I$, where $m=3$ and $n=0$. The blue and red curves  are obtained from (\ref{P_m_n_theorem}) and (\ref{P_m_n_asymptotic}), respectively, when $I=0.4$. The yellow and magenta curves  are obtained from (\ref{P_m_n_theorem}) and (\ref{P_m_n_asymptotic}), respectively, when $I=0.6$. It is seen from Fig. \ref{fig_vs} that the curves obtained from (\ref{P_m_n_theorem}) converge to the corresponding curves obtained from (\ref{P_m_n_asymptotic}) with the same $I$ as $l$ increases. Moreover, the curves obtained from (\ref{P_m_n_asymptotic}) are always above the corresponding curves obtained from (\ref{P_m_n_theorem}), which agrees with Theorem \ref{Theorem_asymptotic}. 

Next, we compare the success probability of a TR-DSA with that of a TU-DSA. As $I$ increases, Fig. \ref{fig_vsl} depicts the success probability of a TR-DSA, $P_s^{(TR)}$, and the success probability of a TU-DSA, $P_s^{(TU)}$, for different $L$, where $Z=4$. The blue dashed curve illustrates the success probability of a TU-DSA $P_s^{(TU)}$. The green curve marked with squares, the purple curve marked with `x's, the red curve marked with diamonds, the yellow curve marked with circles, and the teal curve marked with triangles illustrate the success probabilities of a TR-DSA, $P_s^{(TR)}$, when $L$ equals 1, 2, 10, 20, and 50, respectively. It is seen that $P_s^{(TR)}$ increases as $L$ increases when $0<I<1$, and $P_s^{(TR)}$ is no greater than $P_s^{(TU)}$, which agrees with Corollary \ref{Corollary_1}. In addition, it is seen from Fig. \ref{fig_vsl} that $P_s^{(TU)}$ is one if $I\ge0.5$. However, $P_s^{(TR)}$ is smaller than one when $0.5\le I<1$, which also agrees with Corollary \ref{Corollary_1}. 

\begin{figure}[H]\centering   
	\includegraphics[width=0.5\textwidth]{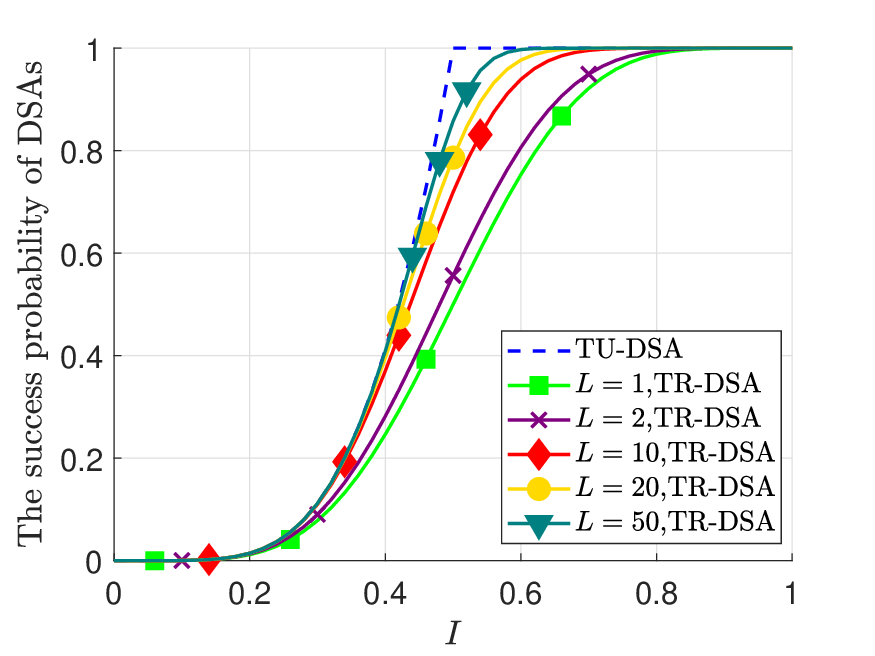}    
	\caption{The comparison between $P_s^{(TR)}$ and $P_s^{(TU)}$  for different $L$.}
	\label{fig_vsl}
\end{figure}
\begin{figure}[H]\centering   
	\includegraphics[width=0.5\textwidth]{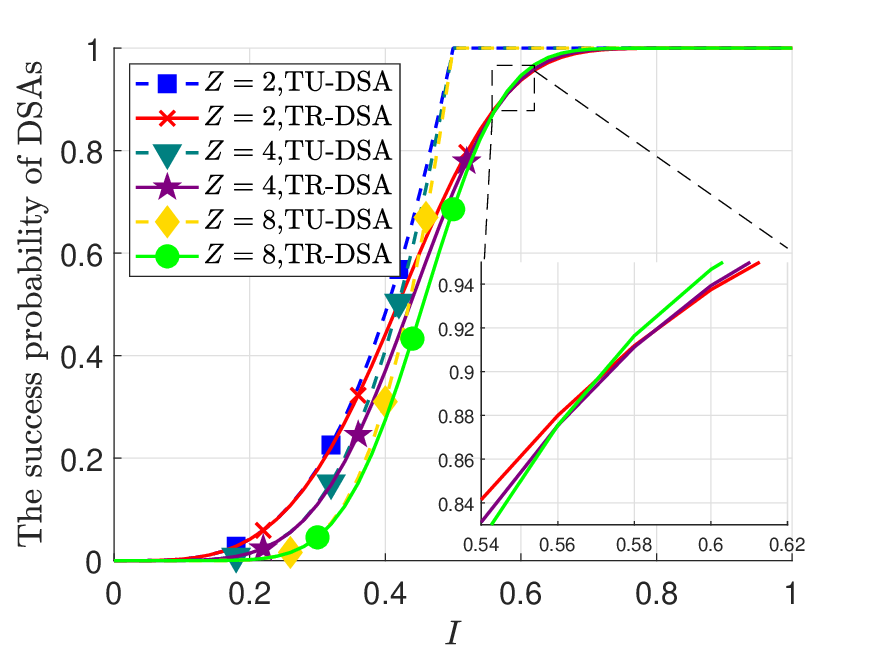}    
	\caption{The comparison between $P_s^{(TR)}$ and $P_s^{(TU)}$ for different $Z$.}
	\label{fig_vsz}
\end{figure}

In the end, we numerically study the impact of $Z$ on $P_s^{(TR)}$ and  $P_s^{(TU)}$.
As $I$ increases, Fig. \ref{fig_vsz} depicts the success probability of a TR-DSA, $P_s^{(TR)}$, and the success probability of a TU-DSA, $P_s^{(TU)}$, for different $Z$, where $L=10$. The blue curve marked with squares, the teal curve marked with triangles, and the yellow curve marked with diamonds illustrate $P_s^{(TU)}$ when the number $Z$ of blocks required for transaction confirmation equals 2, 4, and 8, respectively. The red curve marked with 'x's, the purple curve marked with pentagrams, and the green curve marked with circles illustrate the success probability of a TR-DSA, $P_s^{(TR)}$, when $Z$ equals 2, 4, and 8, respectively. First, it is seen from Fig. \ref{fig_vsz} that, when $I$ is small (i.e., $ I < 0.58$ for $P_s^{(TR)}$ and $I<0.5$ for $P_s^{(TU)}$  in Fig. \ref{fig_vsz}), both $P_s^{(TR)}$ and $P_s^{(TU)}$ decrease as $Z$ increases. This trend aligns with the conventional understanding that a larger number of blocks required for transaction confirmation $Z$ can effectively reduce the risk of a DSA. However, an interesting result is observed when $I$ is large. Specifically, $P_s^{(TR)}$ actually increases as $Z$ increases in Fig. \ref{fig_vsz} when $I$ is greater than 0.58. For example, when $I=0.6$, it is seen from Fig. \ref{fig_vsz1} that  $P_s^{(TR)}$ strictly increases as $Z$ increases. This can be explained by the fact that when $I$ is large, the attacker has a higher probability of mining the next block compared to the honest miners. Therefore, with a larger $Z$, the attacker has more time to secretly mine more blocks for its branch while waiting for transaction confirmation. Consequently, the chance of the attacker's branch surpassing the honest branch during the transaction confirmation period increases. This finding implies that if, in a blockchain network, an attacker gains a majority of the network's hash rate, we can reduce the number of blocks required for transaction confirmation to mitigate the risk of a TR-DSA.

\begin{figure}[H]\centering   
	\includegraphics[width=0.5\textwidth]{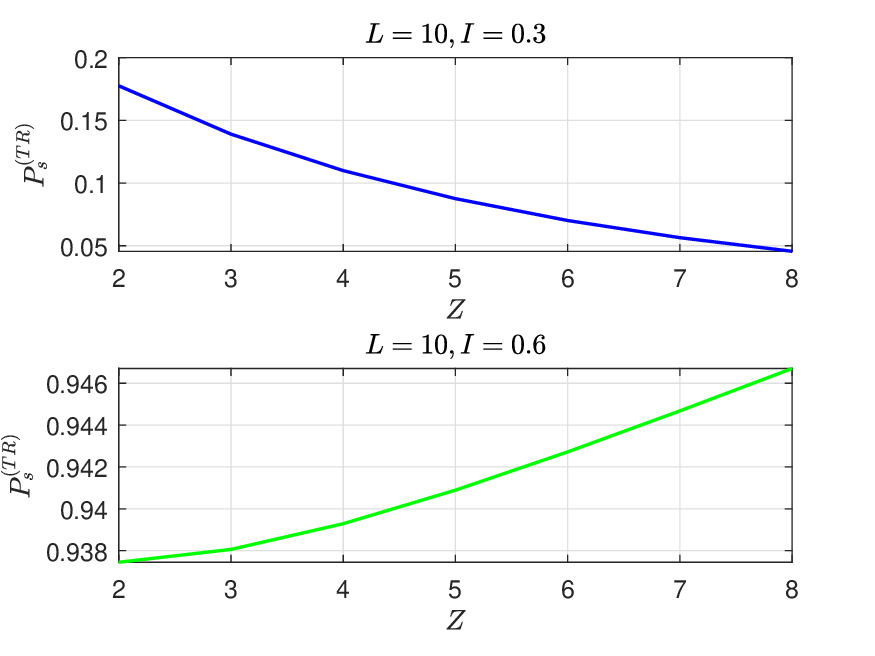}    
	\caption{The relationship between $P_s^{(TR)}$ and $Z$ for different $I$.}
	\label{fig_vsz1}
\end{figure}

\section{Conclusions}\label{Section_conclusion}
In this paper, we have theoretically studied a TR-DSA on a PoW-based blockchain where the attacker mines its own branch within a finite timeframe. This TR-DSA model is well suited for modeling double-spending attacks on blockchain applications where the tasks need to be accomplished within finite timeframes. This TR-DSA model is also appropriate for describing DSAs for the scenarios when a practical attacker does not have sufficient computational resources to conduct a DSA indefinitely.
We have developed a general closed-form expression for the success probability of a TR-DSA. The success probability of a TR-DSA can aid blockchain applications where the tasks need to be accomplished within finite timeframes in evaluating their risks of double-spending attacks, and can be leveraged by a practical attacker with limited computational resources to make an informed decision about whether to proceed with a double-spending attack. Furthermore, we have proven that blockchain applications with timely tasks are less vulnerable to double-spending attacks than blockchain applications which provide attackers with an unlimited timeframe for their attacks. We also have shown that there is still a risk of failure in launching a TR-DSA even if an attacker with limited computational resources amasses a majority of the hash rate in the network.


\appendices



%
%

\ifCLASSOPTIONcaptionsoff
  \newpage
\fi
\bibliographystyle{IEEEtran}
\bibliography{IEEEabrv,Blockchain}

\end{document}